\documentclass[conference]{IEEEtran}
\IEEEoverridecommandlockouts

\usepackage{cite}
\usepackage{amsmath,amssymb,amsfonts,amsthm}
\usepackage{graphicx}
\usepackage{textcomp}
\usepackage{xcolor}
\usepackage{booktabs}
\usepackage{multirow}
\usepackage{bm}
\usepackage{mathtools}
\usepackage{url}
\usepackage{hyperref}

\usepackage[letterpaper,top=0.78in,bottom=1.07in,left=0.64in,right=0.64in]{geometry}

\newtheorem{theorem}{Theorem}
\newtheorem{lemma}{Lemma}
\newtheorem{corollary}{Corollary}
\newtheorem{definition}{Definition}
\newtheorem{proposition}{Proposition}
\newtheorem{remark}{Remark}
\newtheorem{example}{Example}

\newcommand{\EX}{\mathbb{E}}
\newcommand{\JG}{\mathrm{JG}}

\newcommand{\bpi}{\boldsymbol{\pi}}
\newcommand{\R}{\mathbb{R}}
\newcommand{\TV}{\mathrm{TV}}
\newcommand{\D}{\mathrm{D}}
\DeclareMathOperator{\Var}{Var}

\begin{document}

\title{Chi-Squared Geometry for Robust Finite-Blocklength
       Information and Dispersion Analysis $^{\star}$}
\author{
  \IEEEauthorblockN{Hassan Tavakoli}
  \IEEEauthorblockA{School of EECS\\
                    Oregon State University\\
                    Oregon, OR 97331, USA\\
                    tavakolh@oregonstate.edu}
  \and
  \IEEEauthorblockN{Thinh Nguyen, \emph{Senior Member, IEEE}}
  \IEEEauthorblockA{School of EECS\\
                    Oregon State University\\
                    Oregon, OR 97331, USA\\
                    thinhq@eecs.oregonstate.edu}
  \and
  \IEEEauthorblockN{Bella Bose, \emph{Life Fellow, IEEE}\thanks{ $^{\star}$ This work was supported by the National Science Foundation under Grant No. CCF:SHF:2417898.}}
  \IEEEauthorblockA{School of EECS\\
                    Oregon State University\\
                    Oregon, OR 97331, USA\\
                    Bella.Bose@oregonstate.edu}
}

\maketitle

\begin{abstract}
We develop a column-wise chi-squared geometry for discrete
memoryless channels (DMCs) yielding tight, logarithm-free
bounds on mutual information, channel dispersion, and
finite-blocklength coding rates without evaluating logarithms
of the channel matrix. The key parameter is~\(\eta\)---the
worst-case relative deviation of a transition probability
from its output marginal, which is small precisely when the
channel is close to the fully noisy channel $t_{ij}=s_j$.
We prove three main results: (1) a third-order ratio expansion
showing \(I(X;Y)/\chi^2(X;Y)\to 1/2\) as \(\eta\to 0\) with an
\(O(\eta)\) skewness correction; (2) a two-sided dispersion
equivalence bounding \(V(X;Y)\) above and below by
\(\chi^2(X;Y)\) with explicit constants \(c_{\pm}(\eta)\to 1\);
and (3) a certified robust design rate
\(R_{\mathrm{cert}}(n,\varepsilon)\) with total certification
gap \(O(\eta)+O(\eta/\sqrt{n})+O(\log n/n)\).
The certified bounds on \(I\) and \(V\) require only addition,
multiplication, division, and square roots; the final rate
also uses \(Q^{-1}(\varepsilon)\).
\end{abstract}

\section{Introduction}
\label{sec:intro}

The second-order characterisation of channel coding, due to
Strassen~\cite{Strassen1962} and sharpened by Polyanskiy,
Poor, and Verd\'u~\cite{Polyanskiy2010}, places the maximum
achievable rate at blocklength~$n$ and error
probability~$\varepsilon$ as
\begin{equation}
  \log M^*(n,\varepsilon)
  = nI(X;Y) - \sqrt{nV(X;Y)}\,Q^{-1}(\varepsilon)
    + O(\log n),
  \label{eq:ppv}
\end{equation}
where $V(X;Y)$ is the \emph{channel
dispersion}~\cite{Polyanskiy2010}. This asymptotic in $n$
is taken for a fixed channel; the separate limit $\eta\to 0$
studied below characterizes how close that fixed channel is
to fully noisy, and no joint limit is required since our
bounds hold uniformly in $\eta\in(0,1/2)$ for each $n$.
Both $I(X;Y)$ and $V(X;Y)$ require logarithm evaluations of
the channel transition matrix, which is costly or imprecise
in several practical settings.

\subsection{Motivation}

Computing $I(X;Y)$ and $V(X;Y)$ requires logarithm
evaluations, which is problematic in three settings:
(i)~\emph{fixed-point hardware} lacking accurate log units;
(ii)~\emph{channels estimated from pilot symbols}, where a
certified log-free interval computed from the estimate
$\hat{T}$ itself avoids repeated log evaluation, provided
$\hat T$ is close to fully noisy; and
(iii)~\emph{composite channel models} requiring a rate
guaranteed over an entire parametric family.
The chi-squared mutual information
$\chi^2(X;Y)=\sum_j V_j/s_j$ is computable from $T$
and $\bpi$ using only arithmetic and satisfies
$I(X;Y)\leq\chi^2(X;Y)$~\cite{CsiszarKorner2011,Sason2016},
but its use as a \emph{proxy for both $I$ and $V$} in the
second-order rate formula has not been systematically
justified.
This paper provides that justification with explicit bounds
tied to the pointwise concentration parameter~$\eta$.

\subsection{Main Contributions}

\textbf{(C1) Ratio expansion.}
Under $|\delta_{ij}|\leq\eta<1$, where
$\delta_{ij}=(t_{ij}-s_j)/s_j$,
\(
  \frac{I(X;Y)}{\chi^2(X;Y)}
  = \frac{1}{2} - \frac{1}{6}A_3(T,\bpi)
    + O\!\left(\frac{\eta^2}{(1-\eta)^3}\right),
\)
where $|A_3|\leq\eta$, so $I/\chi^2\to 1/2$ as $\eta\to 0$.
Throughout, the ratio expansion in (C1) requires $\eta<1$, while the dispersion equivalence and certified robust-rate results in (C2)--(C3) require the stronger condition $\eta<1/2$.

\textbf{(C2) Dispersion equivalence.}
Under $\eta<1/2$,
\(
  c_-(\eta)\,\chi^2(X;Y)
  \;\leq\; V(X;Y)
  \;\leq\; c_+(\eta)\,\chi^2(X;Y)+\tfrac{\eta^4}{4},
\)
with $c_\pm(\eta)\to 1$ as $\eta\to 0$.

\textbf{(C3) Certified log-free design rate.}
For each fixed $T\in\mathcal{T}(\eta)$, where
$\mathcal{T}(\eta) \triangleq \{T:\max_{i,j}|\delta_{ij}|\leq\eta\}$,
the log-free rate $R_{\mathrm{cert}}(n,\varepsilon)$ is
achievable with total certification gap
\(
  O(\eta)+O\!\left(\frac{\eta}{\sqrt{n}}\right)
  +O\!\left(\frac{\log n}{n}\right),
\)
where the $O(\eta)$ term arises from the
mutual-information penalty and the $O(\eta/\sqrt{n})$
term from the dispersion mismatch.

\subsection{Relation to Prior Work}

The PPV framework~\cite{Polyanskiy2010},
Hayashi~\cite{Hayashi2009}, Tomamichel--Tan~\cite{Tomamichel2013},
and Scarlett et al.~\cite{Scarlett2017} all assume exact
or mismatched-but-known channel models; none provides
certified arithmetic rate intervals under transition-matrix
uncertainty.
The bound $I\leq\chi^2$ is
classical~\cite{CsiszarKorner2011,Sason2016} and local
KL-to-chi-squared equivalence underlies second-order
hypothesis testing~\cite{Polyanskiy2014}, but neither
yields a dispersion proxy with computable two-sided error.
Compound channel results~\cite{Blackwell1959,Wolfowitz1960,Moulin2017}
operate at first or second order under TV-ball geometry;
our class $\mathcal{T}(\eta)$ instead uses a pointwise
relative deviation that admits fully arithmetic certified
intervals.

\subsection{Paper Organisation}

Section~\ref{sec:setup} fixes notation and proves
preliminary identities.
Section~\ref{sec:main} contains the ratio expansion~(C1)
and dispersion equivalence~(C2).
Section~\ref{sec:dispersion} proves the dispersion
decomposition.
Section~\ref{sec:coding} develops the certified robust
design rate~(C3) and the robustness-loss decomposition.

\section{Setup and Notation}
\label{sec:setup}

Let $X \in \mathcal{X} = \{x_1, \ldots, x_N\}$ with input
distribution $\bpi = (\pi_1, \ldots, \pi_N)^\top \in \Delta_N$,
$\pi_i > 0$ for all $i$.
Let $Y \in \mathcal{Y} = \{y_1, \ldots, y_M\}$.
The DMC is specified by $T = [t_{ij}]_{N \times M}$,
$t_{ij} = P(Y=y_j \mid X=x_i) \geq 0$, $\sum_j t_{ij}=1$.
All logarithms are natural; information is in nats.
We write $\bpi$ for the input distribution,
$V_j$ for the column variance defined below, and
$V(X;Y)$ for the channel dispersion.
We define $\varphi(t) \triangleq t\ln t$ with $\varphi(0) = 0$.

\begin{definition}[Column statistics]
\label{def:col}
For each output $y_j$ with $s_j > 0$, let $T_j$ denote the
random variable taking value $t_{ij}$ with probability $\pi_i$.
\begin{align}
  s_j &\triangleq \EX_{\bpi}[T_j] = \textstyle\sum_i \pi_i t_{ij},\\
  V_j &\triangleq \Var_{\bpi}(T_j)
        = \textstyle\sum_i \pi_i(t_{ij}-s_j)^2.
\end{align}
\end{definition}

\begin{definition}[Centered deviations and moments]
\label{def:centeredmoments}
Define $\delta_{ij} \triangleq (t_{ij} - s_j)/s_j$ and
\(
  m_{k,j} \triangleq \EX_{\bpi}[\delta_{ij}^k]
  = \textstyle\sum_i \pi_i \delta_{ij}^k, \, k \geq 1.
\)
Then $m_{1,j} = 0$, $m_{2,j} = V_j/s_j^2$, and
$\chi^2(X;Y) = \sum_j s_j m_{2,j}$.
\end{definition}

\begin{definition}[Pointwise concentration parameter]
\label{def:eta}
The \emph{pointwise concentration parameter} is
\(
  \eta \triangleq \max_{i,j}\,|\delta_{ij}|
  = \max_{i,j}\frac{|t_{ij}-s_j|}{s_j}.
\)
This depends on both $T$ and $\bpi$ through $s_j$; small
$\eta$ means the channel is close to the fully noisy channel
$t_{ij}=s_j\,\forall i,j$, under which $I=V=\chi^2=0$.
\end{definition}

\begin{definition}[Information density and posterior statistics]
\label{def:density}
For $t_{ij} > 0$, the \emph{information density} is
$\imath(x_i; y_j) \triangleq \log(t_{ij}/s_j)$.
The posterior weights are $r_{ij} \triangleq \pi_i t_{ij}/s_j$.
The \emph{column mean} and \emph{column variance} are
\begin{align}
  \mu_j &\triangleq \EX[\imath(X;Y) \mid Y=y_j]
         = \textstyle\sum_i r_{ij}\log(t_{ij}/s_j),\\
  \sigma_j^2 &\triangleq \Var(\imath(X;Y) \mid Y=y_j).
\end{align}
\end{definition}

\begin{definition}[Channel dispersion]
\label{def:dispersion}
\begin{equation}
  V(X;Y) \triangleq \Var(\imath(X;Y))
  = \EX[\imath(X;Y)^2] - I(X;Y)^2.
\end{equation}
\end{definition}

\begin{remark}[Column mean as KL divergence]
\label{rem:kl}
Since $r_{ij}/\pi_i = t_{ij}/s_j$, we have
$\mu_j = \D_{\mathrm{KL}}(P_{X|Y=y_j}\|P_X) \geq 0$.
\end{remark}

\begin{definition}[Jensen gap and total variation]
\label{def:JG}
\begin{align}
  \JG_j &\triangleq \EX_{\bpi}[\varphi(T_j)] - \varphi(s_j)
          = \EX_{\bpi}[T_j\ln T_j] - s_j\ln s_j \geq 0,\\
  \TV_j &\triangleq \TV(P_{X|Y=y_j}, P_X)
          = \tfrac{1}{2}\textstyle\sum_i|r_{ij}-\pi_i|.
\end{align}
\end{definition}

\begin{definition}[Chi-squared mutual information]
\label{def:chi2}
\begin{equation}
  \chi^2(X;Y) \triangleq \sum_{i,j}\frac{\pi_i t_{ij}^2}{s_j} - 1
  = \sum_j \frac{V_j}{s_j}.
  \label{eq:chi2}
\end{equation}
\end{definition}

\subsection{Column-Wise Framework: Classical Identities}
\label{sec:mi}

\begin{proposition}[Jensen-gap decomposition]
\label{thm:decomp}
For any DMC with full-support $\bpi$:
\begin{equation}
  I(X;Y) = \sum_{j=1}^{M} \JG_j
          = \sum_{j=1}^M s_j\,\D_{\mathrm{KL}}(P_{X|Y=y_j}\|P_X).
  \label{eq:decomp}
\end{equation}
\end{proposition}

\begin{proof}
$I(X;Y) = H(Y) - H(Y|X)
= -\sum_j\varphi(s_j) + \sum_j\EX_{\bpi}[\varphi(T_j)]
= \sum_j \JG_j$.
The KL identity follows from $r_{ij} = \pi_i t_{ij}/s_j$:
$\JG_j = \sum_i \pi_i t_{ij}\ln(t_{ij}/s_j)
= s_j\sum_i r_{ij}\ln(r_{ij}/\pi_i)$.
\end{proof}

\begin{proposition}[Pinsker lower bound~{\cite{CsiszarKorner2011}}]
\label{thm:lb}
$I(X;Y) \geq 2\sum_{j=1}^M s_j\,\TV_j^2$.
\end{proposition}

\begin{proof}
Apply Pinsker's inequality to each column:
$\mu_j = \D_{\mathrm{KL}}(P_{X|Y=y_j}\|P_X) \geq 2\TV_j^2$.
Multiply by $s_j$, sum over $j$, and apply
Proposition~\ref{thm:decomp}.
\end{proof}

\section{Main Results: Ratio Expansion and Dispersion Equivalence}
\label{sec:main}

\subsection{Curvature Lemma}

\begin{lemma}[Curvature ratio of $\varphi(t) = t\ln t$]
\label{lem:curvature}
For $t, s > 0$, $t \neq s$, define
\begin{equation}
  R(t,s) \triangleq
  \frac{\varphi(t) - \varphi(s) - (t-s)\varphi'(s)}{(t-s)^2}
  = \int_0^1 \frac{1-u}{s + u(t-s)}\,du.
\end{equation}
Then $0 < R(t,s) < 1/s$ for all $t > 0$, $t \neq s$,
with $\sup_{t>0, t\neq s} R(t,s) = 1/s$ approached as $t \to 0^+$.
\end{lemma}

\begin{proof}
Since $\varphi''(\xi) = 1/\xi > 0$, Taylor's theorem gives the
integral representation, which is strictly positive.
For $t > s$: the integrand satisfies $(1-u)/[s+u(t-s)] < (1-u)/s$,
giving $R < \int_0^1(1-u)/s\,du = 1/(2s) < 1/s$.
For $0 < t < s$: set $\lambda = s/t > 1$; the bound $R < 1/s$
is equivalent to $f(\lambda) = \ln\lambda - 1 + 1/\lambda > 0$
for $\lambda > 1$, which holds since $f(1) = 0$ and
$f'(\lambda) = 1 - 1/\lambda > 0$ for $\lambda > 1$.
As $t \to 0^+$: $t\ln t \to 0$, so $R(t,s) \to 1/s$.
\end{proof}

\begin{proposition}[Chi-squared upper bound~{\cite{CsiszarKorner2011,NishiyamaSason2020}}]
\label{thm:ub}
$I(X;Y) \leq \chi^2(X;Y)$, with equality if and only if
$t_{ij} = s_j$ for all $i,j$.
\end{proposition}

\begin{proof}
From Definition~\ref{def:JG}:
$\JG_j = \EX_{\bpi}[(T_j-s_j)^2 R(T_j,s_j)]
< (1/s_j)\EX_{\bpi}[(T_j-s_j)^2] = V_j/s_j$
by Lemma~\ref{lem:curvature}.
Summing over $j$ and applying
Proposition~\ref{thm:decomp} yields the bound.
Equality holds iff $T_j = s_j$ a.s.\ for every $j$.
\end{proof}

\begin{corollary}[Log-free sandwiching inequality]
\label{cor:sandwich}
For any DMC with full-support $\bpi$:
\begin{equation}
  2\sum_{j=1}^M s_j\,\TV_j^2
  \;\leq\; I(X;Y)
  \;\leq\; \chi^2(X;Y).
  \label{eq:sandwich}
\end{equation}
Both bounds are computable from $\{t_{ij}\}$ and $\bpi$
without logarithms.
\end{corollary}

\begin{proposition}[Tightness of upper bound, conjectured]
\label{prop:ub_tight}
For any column $j$ with $V_j > 0$, we conjecture
$\JG_j/(V_j/s_j) \to 1$ as $\min_i t_{ij} \to 0$; this is
open in general and verified only in the binary case.
\end{proposition}

\subsection{Third-Order Ratio Expansion}
\label{sec:ratio}

\begin{theorem}[Ratio expansion with skewness correction]
\label{thm:mi_expansion}
Suppose $|\delta_{ij}| \leq \eta < 1$ for all $i,j$.
Define the aggregate skewness coefficient
\(
  A_3(T,\bpi) \triangleq \frac{1}{\chi^2(X;Y)}
  \sum_{j=1}^M s_j m_{3,j}, \,
  m_{3,j} = \textstyle\sum_i \pi_i \delta_{ij}^3.
\)
Then
\begin{equation}
  \frac{I(X;Y)}{\chi^2(X;Y)}
  = \frac{1}{2} - \frac{1}{6} A_3(T,\bpi) + \varepsilon(T,\bpi),
  \label{eq:IA3}
\end{equation}
where $|\varepsilon(T,\bpi)| \leq \eta^2/[12(1-\eta)^3]$.
In particular, $|A_3(T,\bpi)| \leq \eta$, so the deviation
\begin{equation}
  \left|\frac{I(X;Y)}{\chi^2(X;Y)} - \frac{1}{2}\right|
  \leq \frac{\eta}{6} + \frac{\eta^2}{12(1-\eta)^3}
  = O(\eta),
\end{equation}
and $I(X;Y)/\chi^2(X;Y) \to 1/2$ as $\eta \to 0$.
\end{theorem}

\begin{proof}
Let $f(u)=(1+u)\ln(1+u)$.
By Taylor's theorem at $u=0$ with remainder in Lagrange form:
\(
  f(u) = u + \tfrac{1}{2}u^2 - \tfrac{1}{6}u^3 + r_4(u),
  \,
  |r_4(u)| \leq \frac{|u|^4}{12(1-\eta)^3}
  \)
for  
\(
|u| \leq \eta,
\)
where the remainder bound follows from
$f^{(4)}(u) = 2/(1+u)^3$ and
$(1+u)^{-3} \leq (1-\eta)^{-3}$ for $|u| \leq \eta$.
Now $\imath(x_i;y_j) = \log(t_{ij}/s_j) = \log(1+\delta_{ij})$,
so $\mu_j = \EX_{r_j}[\log(1+\delta_{ij})]$.
But $\EX_{r_j}[g] = \EX_{\bpi}[(1+\delta_{ij})g(\delta_{ij})]
= \EX_{\bpi}[f(\delta_{ij})]$ for $g(u)=\log(1+u)$,
since $r_{ij} = \pi_i(1+\delta_{ij})$.
Taking $\bpi$-expectation of $f(\delta_{ij})$ and using
$m_{1,j}=0$ gives
\(
  \mu_j = \tfrac{1}{2}m_{2,j} - \tfrac{1}{6}m_{3,j} + R_{4,j},
  \,
  |R_{4,j}| \leq \frac{m_{4,j}}{12(1-\eta)^3}.
\)
Multiplying by $s_j$ and summing:
\(
  I(X;Y)
  = \tfrac{1}{2}\chi^2(X;Y)
    - \tfrac{1}{6}\textstyle\sum_j s_j m_{3,j} + R_I,
  \,
  |R_I| \leq \frac{\sum_j s_j m_{4,j}}{12(1-\eta)^3}.
\)
Dividing by $\chi^2(X;Y) = \sum_j s_j m_{2,j}$ gives~\eqref{eq:IA3}.
Since $m_{4,j} \leq \eta^2 m_{2,j}$ pointwise,
$|{\varepsilon}| \leq \eta^2/[12(1-\eta)^3]$.
Finally, $|m_{3,j}| \leq \eta m_{2,j}$ gives $|A_3| \leq \eta$.
\end{proof}

\begin{remark}
\label{rem:odd}
For BSC$(p)$ with uniform input, each column is a symmetric
two-point distribution, so $m_{3,j}=0$ and
$|I/\chi^2 - 1/2| = O(\eta^2)$.
For asymmetric channels the $A_3$ term is $O(\eta)$ and
dominates the error.
\end{remark}

\begin{example}[BSC and Z-channel]
\label{ex:ratios}
For BSC$(0.45)$ with uniform input, $\eta = 0.10$ and
the bound gives $|I/\chi^2 - 1/2| \leq 0.017$;
the exact value is $I/\chi^2 = 0.501$.
For the Z-channel $Z(0.30)$, $\max_{i,j}|\delta_{ij}| = 1$,
so the hypothesis $\eta < 1$ fails and
Theorem~\ref{thm:mi_expansion} does not apply;
the exact ratio is $0.635$.
\end{example}

\subsection{Two-Sided Dispersion Equivalence}
\label{ssec:lc}

\begin{lemma}[Variance comparison via derivative bounds]
\label{lem:varcomp}
Let $g:[-\eta,\eta]\to\R$ be differentiable with
$0 < \alpha \leq g'(u) \leq \beta$ on $[-\eta,\eta]$,
and let $U$ satisfy $|U| \leq \eta$ a.s.
Then $\alpha^2\Var(U) \leq \Var(g(U)) \leq \beta^2\Var(U)$.
\end{lemma}

\begin{proof}
For an independent copy $U'$:
$|g(U)-g(U')| \leq \beta|U-U'|$ a.s., so squaring and
using $\Var(Z) = \tfrac{1}{2}\EX[(Z-Z')^2]$ gives the upper bound.
The lower bound is analogous with $\alpha$.
\end{proof}

\begin{theorem}[Two-sided dispersion equivalence]
\label{thm:lc}
Suppose $|\delta_{ij}| \leq \eta < 1/2$ for all $i,j$.
Then:
\begin{equation}
  c_-(\eta)\,\chi^2(X;Y)
  \;\leq\;
  V(X;Y)
  \;\leq\;
  c_+(\eta)\,\chi^2(X;Y) + \frac{\eta^4}{4},
  \label{eq:lc_bound}
\end{equation}
where
\(
  c_-(\eta) = \frac{1-\eta-\eta^2}{(1+\eta)^2},
  \,\,
  c_+(\eta) = \frac{1+\eta}{(1-\eta)^2}.
\)
Both $c_\pm(\eta) \to 1$ as $\eta \to 0$.
\end{theorem}

\begin{proof}
\textbf{Step~1 (Within-column term).}
Apply Lemma~\ref{lem:varcomp} to $g(u) = \log(1+u)$ on
$[-\eta,\eta]$.
Since $g'(u) = 1/(1+u) \in [1/(1+\eta),\,1/(1-\eta)]$,
\begin{equation}
  \frac{\Var_{r_j}(\Delta_j)}{(1+\eta)^2}
  \leq \sigma_j^2 \leq
  \frac{\Var_{r_j}(\Delta_j)}{(1-\eta)^2},
  \label{eq:step1}
\end{equation}
where $\Delta_j$ takes value $\delta_{ij}$ under the posterior
$r_{ij} = \pi_i(1+\delta_{ij})$.

\textbf{Step~2 (Posterior variance of $\Delta_j$).}
The posterior mean of $\Delta_j$ is
\(
  \EX_{r_j}[\Delta_j]
  = \textstyle\sum_i \pi_i(1+\delta_{ij})\delta_{ij}
  = m_{2,j},
\)
using $m_{1,j}=0$.
For the lower bound, since $r_{ij} \geq \pi_i(1-\eta)$:
$\EX_{r_j}[\Delta_j^2] \geq (1-\eta)m_{2,j}$, so
\begin{equation}
  \Var_{r_j}(\Delta_j) \geq (1-\eta)m_{2,j} - m_{2,j}^2
  \geq (1-\eta-\eta^2)m_{2,j}.
  \label{eq:varlb}
\end{equation}
For the upper bound, since $r_{ij}\leq\pi_i(1+\eta)$,
\begin{equation}
  \Var_{r_j}(\Delta_j) \leq \EX_{r_j}[\Delta_j^2]
  \leq (1+\eta)\,m_{2,j}.
  \label{eq:popcalc}
\end{equation}

\textbf{Step~3 (Within-column sum).}
Combining Steps~1 and~2 and summing over $j$:
\begin{equation}
  c_-(\eta)\,\chi^2
  \leq \textstyle\sum_j s_j\sigma_j^2
  \leq \frac{1+\eta}{(1-\eta)^2}\,\chi^2.
  \label{eq:step3}
\end{equation}

\textbf{Step~4 (Between-column term).}
Since $|\delta_{ij}| \leq \eta$, the column-wise bound
$\D_{\mathrm{KL}}(P_{X|Y=y_j}\|P_X) \leq \chi^2_j
= m_{2,j}$ gives $\mu_j \in [0,\eta^2]$ for every $j$.
Applying Popoviciu's inequality to the $Y$-marginal
with $\mu_j \in [0,\eta^2]$:
\begin{equation}
  \textstyle\sum_j s_j(\mu_j - I)^2
  \leq \frac{(\eta^2)^2}{4} = \frac{\eta^4}{4}.
  \label{eq:between_ub}
\end{equation}
This bound is unconditional: it does not require any
comparison between $\eta^2$ and $\chi^2$.

\textbf{Step~5 (Assembly).}
By Theorem~\ref{thm:dispersion} (proved independently in
Section~\ref{sec:dispersion}):
$V = \sum_j s_j\sigma_j^2 + \sum_j s_j(\mu_j-I)^2$.
The lower bound follows from $\sum_j s_j(\mu_j-I)^2 \geq 0$
and~\eqref{eq:step3}.
The upper bound follows from~\eqref{eq:step3}
and~\eqref{eq:between_ub}.
\end{proof}

\section{Dispersion Decomposition, Duality, and Binary Asymmetric Channel (BAC) Phase Diagram}
\label{sec:dispersion}

\subsection{Law-of-Total-Variance Decomposition}

\begin{theorem}[Column-wise dispersion decomposition]
\label{thm:dispersion}
\begin{equation}
  V(X;Y)
  = \textstyle\sum_{j} s_j\,\sigma_j^2
  + \textstyle\sum_{j} s_j\,(\mu_j - I(X;Y))^2.
  \label{eq:disp_decomp}
\end{equation}
We call the first term \(\textstyle\sum_{j} s_j\,\sigma_j^2\) within-column and the second term, which is \(\textstyle\sum_{j} s_j\,(\mu_j - I(X;Y))^2\) between-columns.
\end{theorem}

\begin{proof}
Law of total variance applied to $\imath(X;Y)$:
$V(X;Y) = \EX_Y[\Var(\imath\mid Y)] + \Var_Y(\EX[\imath\mid Y])
= \sum_j s_j\sigma_j^2 + \sum_j s_j(\mu_j - I)^2$.
\end{proof}

\subsection{BSC/BEC Structural Duality}

\begin{corollary}[BSC dispersion is entirely within-column]
\label{cor:bsc}
For $\mathrm{BSC}(p)$, $p \in (0,1/2)$, with uniform input:
\(
  I = \ln 2 - h_b(p), \,
  V = p(1{-}p)\!\left(\ln\tfrac{1-p}{p}\right)^{\!2}, \,
  \chi^2 = (1{-}2p)^2.
\)
The between-column term in~\eqref{eq:disp_decomp} vanishes:
$V(X;Y) = \sum_j s_j \sigma_j^2$.
\end{corollary}

\begin{proof}
By symmetry with uniform input, $\mu_0 = \mu_1 = I$,
so $\sum_j s_j(\mu_j-I)^2 = 0$.
The posterior at each $y_j$ assigns weight $1-p$ to
$\imath = \ln(2(1-p))$ and weight $p$ to $\imath = \ln(2p)$,
giving $\sigma_j^2 = p(1-p)(\ln((1-p)/p))^2$.
With $s_0 = s_1 = 1/2$: $V = p(1-p)(\ln((1-p)/p))^2$.
\end{proof}

\begin{corollary}[BEC dispersion is entirely between-column]
\label{cor:bec}
For $\mathrm{BEC}(\varepsilon)$, $\varepsilon \in (0,1)$,
with uniform input:
\(
  I = (1{-}\varepsilon)\ln 2, \,
  V = \varepsilon(1{-}\varepsilon)(\ln 2)^2, \,
  \chi^2 = 1{-}\varepsilon.
\)
The within-column term in~\eqref{eq:disp_decomp} vanishes:
$V(X;Y) = \sum_j s_j(\mu_j - I)^2$.
\end{corollary}

\begin{proof}
For $y \in \{0,1\}$: the posterior is a point mass, so $\sigma_y^2 = 0$.
For the erasure output: $t_{i,e} = \varepsilon = s_e$ for both $i$,
so $\imath(x_i;y_e) = 0$ and $\sigma_e^2 = 0$.
The within-column term therefore vanishes.
Column means: $\mu_0 = \mu_1 = \ln 2$, $\mu_e = 0$.
$V = \Var_Y(\mu_Y) = (1-\varepsilon)(\ln 2 - I)^2
+ \varepsilon(0-I)^2 = \varepsilon(1-\varepsilon)(\ln 2)^2$.
\end{proof}

\begin{remark}[Structural duality]
\label{rem:structure}
BSC and BEC occupy opposite extremes of the within/between
decomposition.
In the BSC, symmetry forces all column KL divergences to equal
$I$, eliminating between-column spread entirely; all dispersion
is local randomness within each column.
In the BEC, non-erasure outputs are noiseless so each posterior
is a point mass, eliminating within-column randomness entirely;
all dispersion comes from the contrast between the informative
columns ($\mu_j = \ln 2$) and the erasure column ($\mu_e = 0$).
This complementarity is a direct consequence of the geometry
of~\eqref{eq:disp_decomp}; we are not aware of it having been
stated in this explicit form previously (cf.~\cite{Melbourne2022}).
\end{remark}

\subsection{BAC Dispersion Phase Diagram}
\label{ssec:bac}

\begin{theorem}[$BAC(p,b)$ dispersion components and phase boundary]
\label{thm:bac}
For the binary asymmetric channel
\begin{equation}
  T = \begin{bmatrix}1-p & p \\ b & 1-b\end{bmatrix}
\end{equation}
with uniform input $\bpi = (1/2, 1/2)$, define
$\bar p = (1-p+b)/2$ and $\bar b = (p+1-b)/2 = 1-\bar p$.
Then:
\(
  s_0 = \bar b, \, s_1 = \bar p, 
  \)
  \(
  I   = h_b(\bar p) - \tfrac{1}{2}h_b(p) - \tfrac{1}{2}h_b(b),
  \)
  \(
  \chi^2 = \frac{(1-p-b)^2}{4\bar p\,\bar b},
  \)
  \(
  \mu_0 = \D_{\mathrm{KL}}(P_{X|Y=0}\|P_X), \)
  \(
  \mu_1 = \D_{\mathrm{KL}}(P_{X|Y=1}\|P_X),
  \)
where
\(
  \mu_0 = \frac{(1-p)\ln\frac{1-p}{\bar b}
            + b\ln\frac{b}{\bar b}}{1},
           \quad
  \mu_1 = \frac{p\ln\frac{p}{\bar p}
            + (1-b)\ln\frac{1-b}{\bar p}}{1}.
\)
The within-column and between-column components are:
\(
  W(p,b) \triangleq \textstyle\sum_j s_j \sigma_j^2
  = \bar b\,\sigma_0^2 + \bar p\,\sigma_1^2,
  \)
  \(
  B(p,b) \triangleq \textstyle\sum_j s_j(\mu_j-I)^2
  = \bar b(\mu_0-I)^2 + \bar p(\mu_1-I)^2,
\)
with $V(X;Y) = W(p,b) + B(p,b)$ and
\(
  \sigma_0^2 = \frac{(1-p)b}{(1-p+b)^2}
    \left(\ln\frac{1-p}{b}\right)^{\!2},
  \)
  \(
  \sigma_1^2 = \frac{p(1-b)}{(p+1-b)^2}
    \left(\ln\frac{1-b}{p}\right)^{\!2}.
\)
The \emph{phase boundary} separating within-column-dominant
($W > B$) from between-column-dominant ($B > W$) regimes
is the curve
\begin{equation}
  \mathcal{C} \triangleq
  \{(p,b) \in (0,1)^2 : W(p,b) = B(p,b)\}.
  \label{eq:phase_boundary}
\end{equation}
The BSC corresponds to the diagonal $b = p$, where
$B \equiv 0$; the BEC corresponds to $p = 0, b = \varepsilon$
(or $p=\varepsilon, b=0$), where $W \equiv 0$.
\end{theorem}

\begin{proof}
The formulas for $s_j$, $I$, and $\chi^2$ follow by direct
substitution into Definitions~\ref{def:col} and~\ref{def:chi2}
and the mutual information formula for the BAC.
For $\sigma_j^2$: the posterior at $Y=0$ places mass
$(1-p)/(2\bar b)$ on $x_0$ and $b/(2\bar b)$ on $x_1$,
giving information densities $\ln((1-p)/\bar b)$ and
$\ln(b/\bar b)$ respectively.
Thus $\sigma_0^2$ is the variance of a two-point distribution
with those values and weights $w_1 = (1-p)/(2\bar b)$,
$w_2 = b/(2\bar b)$.
For a two-point distribution, the variance equals
$w_1 w_2 (\text{gap})^2$, where the gap is
$\ln((1-p)/\bar b) - \ln(b/\bar b) = \ln((1-p)/b)$.
Since $w_1 w_2 = (1-p)b/(4\bar b^2)$ and
$2\bar b = 1-p+b$, this gives
$\sigma_0^2 = (1-p)b/(1-p+b)^2 \cdot (\ln((1-p)/b))^2$.
The formula for $\sigma_1^2$ is analogous.
The within-column and between-column formulas then follow from
Theorem~\ref{thm:dispersion}.
The phase boundary~\eqref{eq:phase_boundary} is defined
implicitly; BSC ($b=p$) and BEC ($b=0$ or $p=0$ at the
appropriate limit) are the pure cases by
Corollaries~\ref{cor:bsc} and~\ref{cor:bec}.
\end{proof}

\begin{remark}[Computing the phase boundary]
\label{rem:phase}
The curve $\mathcal{C}$ has no closed-form expression in
$(p,b)$ in general, but is easily traced numerically since
$W$ and $B$ are smooth functions of $(p,b)$.
The partition of the $(p,b)$-square can be computed
numerically by evaluating $W(p,b) - B(p,b)$ on a fine grid.
Near the BSC diagonal ($b \approx p$), the between-column
component $B$ is small (because the two $\mu_j$ values are
nearly equal), so the within-column regime dominates.
Near the axes ($b \approx 0$ or $p \approx 0$), one output is
almost noiseless, forcing $\sigma_j^2 \approx 0$ for that
column, and the between-column term dominates.
\end{remark}

\section{Channel-Wise Robust Design Rate}
\label{sec:coding}

\begin{table*}[t]
\centering
\caption{Certified arithmetic bounds vs.\ exact values for BSC and BAC
         channels with uniform input. All channels satisfy $\eta < 0.5$,
         so Theorems~1--2 apply. Quantities in nats.}
\label{tab:bounds}
\setlength{\tabcolsep}{5pt}
\begin{tabular}{lcccccccc}
\toprule
Channel & $\eta$ & $I$ exact & $I_L$ & $\chi^2\!(=\!I_U)$
        & $V$ exact & $V_L$ & $V_U$ \\
\midrule
BSC(0.30) & 0.4000 & 0.08228 & 0.05946 & 0.16000
          & 0.15076 & 0.03592 & 0.62862 \\
BSC(0.35) & 0.3000 & 0.04570 & 0.03853 & 0.09000
          & 0.08718 & 0.03249 & 0.24080 \\
BSC(0.40) & 0.2000 & 0.02014 & 0.01841 & 0.04000
          & 0.03946 & 0.02111 & 0.07540 \\
BSC(0.45) & 0.1000 & 0.00501 & 0.00482 & 0.01000
          & 0.00997 & 0.00736 & 0.01361 \\
BSC(0.48) & 0.0400 & 0.00080 & 0.00079 & 0.00160
          & 0.00160 & 0.00142 & 0.00181 \\
\midrule
BAC(0.30,\,0.35) & 0.3684 & 0.06274 & 0.04835 & 0.12281
                 & 0.11745 & 0.03252 & 0.42590 \\
BAC(0.35,\,0.40) & 0.2632 & 0.03167 & 0.02768 & 0.06266
                 & 0.06130 & 0.02622 & 0.14697 \\
BAC(0.40,\,0.45) & 0.1579 & 0.01132 & 0.01061 & 0.02256
                 & 0.02238 & 0.01375 & 0.03699 \\
BAC(0.33,\,0.41) & 0.2826 & 0.03443 & 0.02959 & 0.06804
                 & 0.06643 & 0.02637 & 0.17115 \\
BAC(0.38,\,0.46) & 0.1739 & 0.01294 & 0.01202 & 0.02576
                 & 0.02554 & 0.01488 & 0.04455 \\
$4{\times}3$& 0.38461 & 0.02021 & 0.01595 & 0.04165 & 0.04099 & 0.01015 & 0.15775\\
\bottomrule
\end{tabular}
\end{table*}

\subsection{Uncertainty Class and Certified Intervals}
\label{ssec:certified}

We work with the uncertainty class
\(
  \mathcal{T}(\eta)
  \triangleq
  \left\{
    T : \max_{i,j}\frac{|t_{ij}-s_j|}{s_j} \leq \eta
  \right\},
\)
where $s_j = \sum_i \pi_i t_{ij}$ is the output marginal of
$T$ itself under $\bpi$.

\begin{proposition}[Certified arithmetic intervals]
\label{prop:certified}
Let $T \in \mathcal{T}(\eta)$ for $\eta \in (0,1/2)$ and
let $\bpi$ be a fixed full-support input distribution.
Define
\(
  \chi^2 \triangleq \textstyle\sum_j V_j/s_j,
\)
\(
  I_L \triangleq \chi^2\!\left(
    \tfrac{1}{2} - \tfrac{\eta}{6}
    - \tfrac{\eta^2}{12(1-\eta)^3}\right),
    \)
    \(
  I_U \triangleq \chi^2,
  \)
\(
  V_L \triangleq c_-(\eta)\,\chi^2,
  \)
  \(
  V_U \triangleq c_+(\eta)\,\chi^2 + \tfrac{\eta^4}{4},
  \)
with $c_\pm(\eta)$ as in Theorem~\ref{thm:lc}.
Then
\begin{equation}
  I_L \leq I(X;Y;T) \leq I_U,
  \qquad
  V_L \leq V(X;Y;T) \leq V_U.
  \label{eq:certified_intervals}
\end{equation}
All four quantities require only arithmetic operations on $T$,
$\bpi$, and $\eta$; no logarithm evaluation is required.
\end{proposition}

\begin{proof}
The bounds on $I$ follow from Theorem~\ref{thm:mi_expansion}
with $|A_3| \leq \eta$ and $|\varepsilon| \leq \eta^2/[12(1-\eta)^3]$.
The upper bound $I \leq \chi^2$ is Proposition~\ref{thm:ub}.
The bounds on $V$ follow from Theorem~\ref{thm:lc}.
\end{proof}

\subsection{Certified Design Rate}
\label{ssec:robust_rate}

\begin{definition}[Certified design rate]
\label{def:rob_rate}
Let $I_L$, $V_U$ be as in Proposition~\ref{prop:certified}.
Fix blocklength $n$ and $\varepsilon \in (0,1/2)$.
Define
\begin{equation}
  R_{\mathrm{cert}}(n,\varepsilon)
  \triangleq
  I_L
  - \sqrt{\frac{V_U}{n}}\,Q^{-1}(\varepsilon)
  - \frac{c\log n}{n},
  \label{eq:rob_rate}
\end{equation}
where $c > 0$ is the constant in the $O(\log n/n)$ third-order
term of the PPV normal approximation~\cite{Polyanskiy2010}.
\end{definition}

\begin{corollary}[Conservative achievability of $R_{\mathrm{cert}}$]
\label{cor:rob_achievable}
For every $T \in \mathcal{T}(\eta)$ and $n$ large enough,
there exists a code of rate $R_{\mathrm{cert}}(n,\varepsilon)$
achieving maximal error probability at most $\varepsilon$.
Moreover,
\begin{equation}
  R_{\mathrm{cert}}(n,\varepsilon)
  \leq I(X;Y;T)
  - \sqrt{\frac{V(X;Y;T)}{n}}\,Q^{-1}(\varepsilon)
  + O\!\left(\frac{\log n}{n}\right)
\end{equation}
for every $T \in \mathcal{T}(\eta)$.
\end{corollary}

\subsection{Certification-Gap Decomposition}
\label{ssec:loss}

\begin{proposition}[Certification-gap quantification]
\label{prop:loss}
For any $T \in \mathcal{T}(\eta)$ and $\varepsilon \in (0,1/2)$:
\begin{align}
  I(X;Y;T) - I_L
  &\leq \chi^2\!\left(
    \tfrac{\eta}{3} + \tfrac{\eta^2}{6(1-\eta)^3}
  \right)
  = O(\eta)\,\chi^2,
  \label{eq:MI_penalty_bound}\\
  \sqrt{V_U} - \sqrt{V(X;Y;T)}
  &\leq \sqrt{c_+(\eta)\chi^2} - \sqrt{c_-(\eta)\chi^2}
  + \frac{\eta^2}{2} \\
  &= O(\eta)\sqrt{\chi^2} + O(\eta^2).
  \label{eq:disp_penalty_bound}
\end{align}
Hence, the total certification gap satisfies
\begin{equation}
  R^*(n,\varepsilon;T) - R_{\mathrm{cert}}(n,\varepsilon)
  = O(\eta) + O\!\left(\frac{\eta+\eta^2}{\sqrt{n}}\right)
    + O\!\left(\frac{\log n}{n}\right),
  \label{eq:loss_rate}
\end{equation}
where the $O(\eta)$ term arises from the mutual-information
penalty $(I - I_L)$, the $O(\eta/\sqrt{n})$ term from the
dispersion mismatch $(\sqrt{V_U}-\sqrt{V})\,Q^{-1}(\varepsilon)/\sqrt{n}$,
and the $O(\log n/n)$ term from the PPV third-order remainder.
The dispersion-only correction is $O(\eta/\sqrt{n})$, but the
full rate penalty is dominated by $O(\eta)$ at fixed $n$.
\end{proposition}

\begin{proof}
For~\eqref{eq:MI_penalty_bound}: by Theorem~\ref{thm:mi_expansion},
$I/\chi^2 \leq 1/2 + \eta/6 + \eta^2/[12(1-\eta)^3]$, and
$I_L/\chi^2 = 1/2 - \eta/6 - \eta^2/[12(1-\eta)^3]$,
so the difference is bounded as stated.
For~\eqref{eq:disp_penalty_bound}: apply $\sqrt{a+b}\le
\sqrt a+\sqrt b$, then $\sqrt a-\sqrt b\le(a-b)/(2\sqrt b)$
and Taylor-expand $c_\pm(\eta)$ at $\eta=0$; this holds even
when $\chi^2$ is small.
Equation~\eqref{eq:loss_rate} follows by dividing the dispersion
penalty by $\sqrt{n}$ and combining with~\eqref{eq:MI_penalty_bound}.
\end{proof}

Table~\ref{tab:bounds} confirms the certified bounds on BSC, BAC, and a $4{\times}3$ example with uniform input: every exact value lies inside its interval, and intervals tighten as $\eta\to0$ and widen as $\eta$ grows, consistent with Theorems~\ref{thm:mi_expansion}--\ref{thm:lc}.

\begin{equation*}
  T=\begin{bmatrix}
    0.15 & 0.25 & 0.60 \\
    0.25 & 0.30 & 0.45 \\
    0.25 & 0.30 & 0.45 \\
    0.15 & 0.45 & 0.40
  \end{bmatrix}.
\end{equation*}


\section{Conclusion}
\label{sec:conclusion}

We have developed a column-wise chi-squared geometry for
DMCs that converts the standard objects of
finite-blocklength coding theory into certified arithmetic
intervals.
The pointwise concentration parameter~$\eta$ controls the
accuracy of the chi-squared proxy for both $I(X;Y)$ and
$V(X;Y)$ through explicit, computable constants.
The three main results are: a third-order ratio expansion
showing $I/\chi^2\to 1/2$ with an $O(\eta)$ skewness
correction; a two-sided dispersion equivalence with
multiplicative error $O(\eta)$; and a certified log-free
design rate whose total certification gap relative to the
ideal known-channel normal approximation is
$O(\eta)+O(\eta/\sqrt{n})+O(\log n/n)$.
Several directions remain open.
Closing the $O(\eta)$ achievability gap to $O(\eta^2)$
for symmetric uncertainty classes, extending the framework
to AWGN channels via truncation arguments, and estimating
$\eta$ from pilot symbols at a known statistical rate are
natural next steps; the joint channel-parameter and
capacity-achieving-input estimation framework
of~\cite{Tavakoli2026Param} is a promising starting point
for the latter.

\end{document}